\documentclass[conference,letterpaper]{IEEEtran}
\usepackage[utf8]{inputenc} 
\usepackage[T1]{fontenc}
\usepackage{url}
\usepackage{ifthen}
\usepackage{cite}

\IEEEoverridecommandlockouts

\usepackage{multicol,lipsum}

\usepackage[tbtags]{amsmath} 
\usepackage{amssymb}  
\usepackage{verbatim} 
\usepackage{amsxtra}  
\usepackage{comment}

\usepackage{bbm}

\usepackage{enumerate}

\usepackage{amsfonts}
\usepackage{color}

\usepackage{subfig}

\usepackage{wasysym}

\usepackage{cite}

\usepackage{amsthm}

\newtheorem{thm}{Theorem}
\newtheorem{lem}{Lemma}
\newtheorem{cor}{Corollary}

\theoremstyle{definition}
\newtheorem*{defn*}{Definition}
\newtheorem*{scheme*}{Scheme}

\theoremstyle{remark}
\newtheorem{remark}{Remark}

\usepackage{fdsymbol}

\usepackage{tikz}
\usetikzlibrary{shapes,arrows,decorations.markings,positioning,calc}

\tikzstyle{plant} = [draw, fill=red!5, rectangle, 
    minimum height=3em, minimum width=6em]
\tikzstyle{block} = [draw, fill=blue!5, rectangle, 
    minimum height=3em, minimum width=18mm]
\tikzstyle{sum} = [draw, fill=yellow!10, circle, node distance=1cm]
\tikzstyle{coord} = [coordinate]
\tikzstyle{gain} = [draw, fill=red!5, regular polygon, regular polygon sides=3, shape border rotate=-90]
\tikzstyle{pinstyle} = [pin edge={to-,thick,black}]

\tikzstyle{BitPipe} = [thick, decoration={markings,mark=at position
   1 with {\arrow[semithick]{open triangle 60}}},
   double distance=1.4pt, shorten >= 5.5pt,
   preaction = {decorate},
   postaction = {draw,line width=1.4pt, white,shorten >= 4.5pt}]

\usetikzlibrary{fit,backgrounds}

\usepackage{hyperref}

\usepackage{epsfig, graphicx, psfrag}

\usepackage[mathcal]{euscript}
\usepackage[cal=boondox]{mathalfa}

\providecommand{\thmref}[1]{Theorem~\ref{#1}}

\providecommand{\secref}[1]{Sec.~\ref{#1}}

\providecommand{\figref}[1]{Fig.~\ref{#1}}

\usepackage{refcount}

\newcommand{\ceil}[1]{\lceil{#1}\rceil}

\providecommand{\Exp}[1]{{\operatorname{exp} \left\{ #1 \right\}}}

\newcommand{\bm}[1]{\mbox{\boldmath{$#1$}}}

\newcommand{\Comment}[1]{}
\newcommand{\old}[1]{}
\newcommand{\rem}[1]{}

\providecommand{\comment}[1]{}

\newcommand{\beqn}[1]{\begin{eqnarray}\label{#1}}
\newcommand{\eeqn}{\end{eqnarray}}
\newcommand{\beq}[1]{\begin{equation}\label{#1}}
\newcommand{\eeq}{\end{equation}}

\makeatletter
\newcommand{\vast}{\bBigg@{4}}
\newcommand{\Vast}{\bBigg@{5}}
\makeatother

\providecommand{\Exp}[1]{\ensuremath{{\operatorname{exp} \left\{ #1
\right\}}}}

\newcommand{\Dsub}[2]{{#1}_{_{#2}}}

\usepackage{amsmath}

\usepackage{graphicx} 
\usepackage{algorithm}
\usepackage{algpseudocode}
\usepackage{hhline}
\usepackage{amsmath}
\usepackage{bbm}
\usepackage{comment}

\newcommand{\Tal}[1]{{\color{red} #1}}
\newcommand{\Uri}[1]{{\color{blue} #1}}

\begin{document}

\title{Extension of the Poltyrev Bound to Binary Memoryless Symmetric Channels}



\author{Tal Philosof, Ariel Doubchak, Amit Berman and Uri Erez 

\thanks{Tal Philosof, Ariel Doubchak and Amit Berman are with Samsung Semiconductor Research and Development Center, Tel-Aviv, Israel (e-mail: \texttt{\{tal.philosof, ariel.d, amit.berman\}@samsung.com}).}
\thanks{Uri Erez is with the School of Electrical Engineering, Tel Aviv University, Tel Aviv~6997801, Israel (e-mail: \texttt{\{uri\}@eng.tau.ac.il}).}
}

\maketitle

\begin{abstract}
The Poltyrev bound provides a very tight upper bound on the decoding error probability when using binary linear codes for transmission over the binary symmetric channel and the additive white Gaussian noise channel, making use of the code's weight spectrum. In the present work, the bound is extended to  memoryless symmetric channels with a discrete output alphabet. The derived bound is demonstrated on a hybrid BSC-BEC channel. Additionally, a reduced-complexity bound is introduced at the cost of some loss in tightness.

\end{abstract}
\section{Introduction}\label{sec:intro}

The problem of deriving tight upper bounds for the decoding error probability of block codes, and in particular of linear binary block codes, has been widely studied. 
Many of the tighter bounds are variations on Gallager's first bounding technique \cite{gallager1963low} where one partitions the output space into two sets regions (sets). One region captures the event that corresponds to ``bad" channel instantiation. The probability of this event is readily bounded. Over the complement set, a union bound is employed, assuming the weight spectrum of the code is known. We refer the reader to \cite{shamai2002variations,hof2009performance} for a comprehensive review of such bounds.

One of the tightest bounds within this class is due to Poltyrev \cite{PoltyrevBound94}. More precisely, in the latter work, two bounds were derived using a similar method, one tailored to the binary symmetric channel (BSC), the other to the additive white Gaussian noise channel (the ``tangential sphere bound'' (TSB)). 

One may view Poltyrev's bound for the BSC as an instance of Gallager's technique where two key properties of the BSC are employed to obtain an explicit tight bound. The first is that since the output alphabet is binary, and the channel is memoryless and symmetric, it suffices to account for all possible channel output vectors based only on their weight. The second is that the optimal  (or near optimal, in the case of the TSB) choice of region may be explicitly characterized via a single parameter.

In the present work, Poltyrev's bounding method is extended to general binary memoryless symmetric (BMS) discrete-output channels, employing the method of types as the natural extension of Hamming weight. The derived bounds are of practical relevance to scenarios where the output alphabet size is small and/or the block length is small to moderate.

Such a scenario is typical in 
soft-decision decoding as employed in flash memory devices \cite{InsideNAND_book,3DFlash_book}, where the read operation is an expensive resource, leading to coarse quantization \cite{Kurkoski_TIT_2014}. Furthermore, when assessing coding schemes for such devices, in order to obtain robust performance,  oftentimes the channel is modeled as symmetric. 
Thus, the bounds developed in the present work may be applied as a performance metric for  flash memory devices, for block codes whose Hamming weight distributions are known or for code ensembles having a well-characterized average weight distribution.

The derived bound is demonstrated for BCH codes over  a hybrid BSC-BEC ternary-output channel.




\subsection{Notation and Preliminaries}


Consider a BMS channel with conditional probability distribution  
\begin{align}
P_{Y|X}(y|x),
 \label{eq:model_BMS}
\end{align}
where $X \in \mathcal{X} = \{0,1\}$  and $Y \in \mathcal{Y}=\{-M,\ldots,0,\ldots,M\}$ is the channel output with cardinality $2M+1$ as illustrated in \figref{fig:binary_input_sym_output}.
\begin{figure}[htbp]
    \centering
        \includegraphics[scale=0.25
        ]{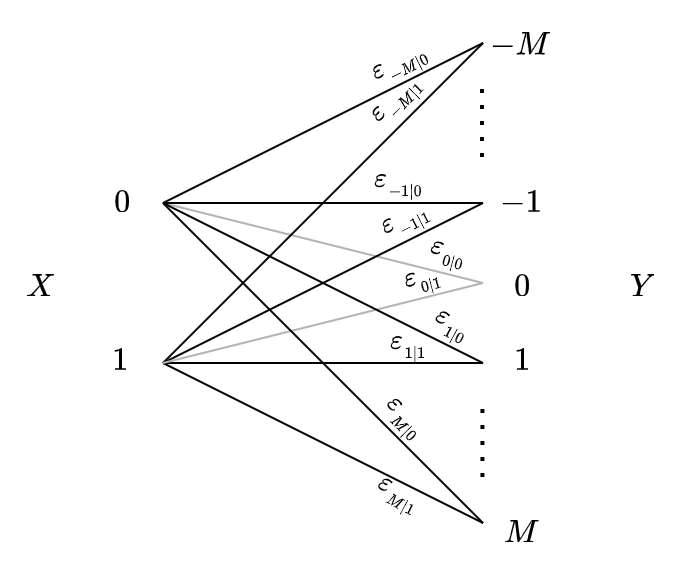}
        \caption{Binary-input symmetric-output channel.}
        \label{fig:binary_input_sym_output}
\end{figure}
The channel is symmetric in the sense that  
\begin{align}
    P_{Y|X}(y|1) = P_{Y|X}(-y|0).
    \label{eq:BMS}
\end{align} 
We denote $P_{Y|X}(y|x) =\varepsilon_{_{y|x}}$ for brevity.

The method of types will be extensively used. 
We briefly recall some pertinent definitions and notations. 
The reader is referred to \cite[Ch. 12]{CoverBook2Edition} for a comprehensive treatment.

The type 
of a vector $\mathbf{y}=y_1^n \in \mathcal{Y}^n$ 
is the relative proportion of occurrences of each symbol.
The set of vectors of type $P_n$
is referred to as the type class of $P_n$, and is denoted by $T(P_n) =\{y_1^n : y_1^n\in P_n\}$. Its size is denoted by $|T(P_n)|$. 
Further, $P_n(j)$ is the number of relative occurrences of symbol $j$
in type $P_n$. We use the notation $\mathcal{P}_n$ to denote a set of types of $y_1^n$, and its complement by $\mathcal{P}_n^c$. Thus, $\mathcal{P}_n \cup \mathcal{P}_n^c \triangleq \mathbbm{P}_n$
contains all possible types of the vectors $\mathbf{y} \in \mathcal{Y}^n$.

We denote the number of occurrences of symbol $a$ in a vector $\mathbf{z}$ by $N(a|\mathbf{z})$. Similarly, we define the number of occurrences of symbol $a$ in a type $P_n$ by 
\begin{align}
    N(a|P_n) = nP_n(a)\in\mathbb{Z}_{+},
    \label{eq:PtoN}
\end{align}
where $\mathbb{Z}_{+}$ indicates the non-negative integers.
We refer to a vector $nP_n=[N(a|P_n(a)]_{a=-M}^M$ is the \emph{non-normalized} counterpart of type $P_n$.

A linear code $\mathcal{C}$ of length $n$, with Hamming weight distribution $S_w,\; w=0,1,\ldots,n$, is assumed to be used for transmission over the channel.
Due to the linearity of the code, and since the channel is symmetric, the error probability is independent of the codeword sent. Hence, we assume, without loss of generality, that the all-zero codeword was transmitted. 



\section{Revisiting the Poltyrev Bound for the BSC}\label{sec:Poltyrev_BSC}

We now recall the Poltyrev bound for a memoryless BSC.  
Here, $M=1$, i.e., $Y \in \{-1,0,1\}$ where the symbol $0$ is received with zero probability and the channel is parameterized by the
crossover probability $\varepsilon$. Specifically, the channel law is 
$
    Y=2({X \oplus Z})-1,
$
where $Z = \rm{Bernoulli}(\varepsilon)$.




Denote the error event by $\mathcal{E}$ and the pairwise error w.r.t. codeword $\mathbf{c}$ by $\mathcal{E}_{c}$.
Building on Gallager's first bounding technique \cite{gallager1963low}, Poltyrev  derived the following upper bound \cite{PoltyrevBound94} for the error probability 
\begin{align}
    P_e &= P(\mathcal{E}, N(1|\mathbf{z}) < \zeta \mid \mathbf{0}) + P(\mathcal{E}, N(1|\mathbf{z}) \geq \zeta \mid \mathbf{0}) \\
    &\leq P(\mathcal{E}, N(1|\mathbf{z}) < \zeta \mid\mathbf{0}) + P(N(1|\mathbf{z}) \geq \zeta)\\
    &=\sum_{\ell=0}^{\zeta-1} P(\mathcal{E}, N(1|\mathbf{z}) = \ell\mid\mathbf{0}) + P(N(1|\mathbf{z}) \geq \zeta)\\
    &=\sum_{\ell=0}^{\zeta-1} P\left(\bigcup_{\mathbf{c}\in \mathcal{C} \setminus \{\mathbf{0}\}} \mathcal{E}_c, N(1|\mathbf{z}) = \ell \middle| \mathbf{0} \right) + P(N(1|\mathbf{z}) \geq \zeta)\\
    &\leq \sum_{\ell=0}^{\zeta-1} \sum_{\mathbf{c}\in \mathcal{C} \setminus \{\mathbf{0}\}} P\left(\mathcal{E}_c, N(1|\mathbf{z}) = \ell\mid \mathbf{0}\right) + P(N(1|\mathbf{z}) \geq \zeta)\\
    &\leq \sum_{\ell=0}^{\zeta-1} \sum_{w=d_{\rm{min}}}^{2\ell} S_w \sum_{\mu=t_w}^{\min\{\ell,w\}} \binom{w}{\mu} \binom{n-w}{\ell-\mu} \varepsilon^{\ell}(1-\varepsilon)^{n-\ell} 
    \nonumber \\
    &\qquad\qquad + \sum_{\ell = \zeta} ^ n \binom{n}{\ell} \varepsilon^{\ell}(1-\varepsilon)^{n-\ell},
\end{align}
where $\ell$ is the number of flipped bits, $1 \leq \zeta \leq n$ and $t_w = \ceil{w/2}$ and where the conditioning on $\mathbf{c}=\mathbf{0}$ is denoted by $\mid \mathbf{0}$ for short.

The bound can be further optimized such that the number of flipped bits $\zeta$ brings the bound to a minimum. Hence, Poltyrev's bound states that $P_e$ is upper bounded by
\begin{align}
   \sum_{\ell=0}^{n} \varepsilon^{\ell}(1-\varepsilon)^{n-\ell}\min\left\{\sum_{w=d_{\rm{min}}}^{2\ell} S_w \sum_{\mu=t_w}^{\min\{\ell,w\}} \binom{w}{\mu} \binom{n-w}{\ell-\mu}, \binom{n}{\ell} \right\}.\label{eq:BSC_Poltyrev_bound} 
\end{align}

The minimizing number of flipped bits, namely $\zeta_0$, is the smallest integer $\ell$ which satisfies the following inequality 
\begin{align}
    \sum_{w=d_{\rm{min}}}^{2\ell} S_w \sum_{\mu=t_w}^{\min\{\ell,w\}} \binom{w}{\mu} \binom{n-w}{\ell-\mu}  \geq \binom{n}{\ell}. 
\end{align}
The bound is characterized by the following:
\begin{itemize}
    \item The summation over the output space reduces to enumeration over the weight of the received vector.
    \item The region is explicitly characterized as a Hamming sphere of radius $\ell$, the latter being a parameter optimized.
\end{itemize}

\section{Extension of the Poltyrev Bound to General Binary-Input Symmetric-Output Channels}\label{sec:extended_Poltyrev}

Consider a BMS channel as defined in \eqref{eq:BMS}.
Let us partition the output space into two regions: vectors with type $P_n \in \mathcal{P}_n$ and their complement, where $\mathcal{P}_n$ is a set of types to be chosen. 
The complement set, all vectors $\mathbf{y}$ with type $P_n \in \mathcal{P}_n^c$, will be referred to as the \emph{large-noise region}.


In order to formulate an upper bound for the probability of error, we further define the following set of types corresponding to pairwise error events
\begin{align}
    &\mathcal{P}_w^\mathcal{E} \triangleq 
    \left\{ P_w: \frac{P_{\mathbf{Y}|\mathbf{X}}(\mathbf{y}_w|\mathbf{x}_w = \mathbf{0})}{ P_{\mathbf{Y}|\mathbf{X}}(\mathbf{y}_w|\mathbf{x}_w =\mathbf{1})}\leq 1,\;  \mathbf{y}_w \in P_w\right\}, 
    \label{eq:pairwise_error_event_region}
\end{align}
where $\mathbf{y}_w = y_1^w$ and $\mathbf{x}_w = x_1^w$.

As we have assumes the channel is memoryless, this set takes the form
\begin{align}
    \mathcal{P}_w^\mathcal{E} =\left\{ P_w : \prod_{j=-M}^M  \left[\frac{\Dsub{\varepsilon}{j|0}}{\Dsub{\varepsilon}{j|1}}\right]^{w{P_w}(j)} \leq 1 \right\} \label{eq:error_types}.
\end{align}
Alternatively, in terms of the \emph{Log-Likelihood Ratio} and using definition \eqref{eq:PtoN}, the set can be written as
\begin{align}
    \mathcal{P}_w^\mathcal{E} =\left\{P_w: \sum_{j=-M}^M N(j|P_w)\cdot \rm{LLR}_j \leq 0\right\} \label{eq:error_types_LLR},
\end{align}
where $\rm{LLR}_j \triangleq \log\left(\frac{\Dsub{\varepsilon}{j|0}}{\Dsub{\varepsilon}{j|1}}\right)$. The set of all possible non-normalized types of the output vector $\mathbf{y} \in \mathcal{Y}^n$ is given by
\begin{align}
n\mathbbm{P}_n = \left\{\bm{\ell}: \sum_{j=-M}^M \ell_j = n,\;\ell_j \in \mathbbm{Z}_{+} \right\}.
\end{align}
Finally, define  the set
of all possible ``conditional non-normalized types" of an output subvector $\mathbf{y}_w$ given the output vector $\mathbf{y}$ is of non-normalized type $\bm{\ell}$, as follows:  
\begin{align}
        \mathcal{U}_w(\mathbf{\bm{\ell}}) &\triangleq \Bigg\{ \mathbf{\bm{\mu}} : \sum_{i=-M}^M \mu_i=w,\;\mu_j\in\mathbb{Z}_{+}  \nonumber \\
          &\qquad ,0 \leq \mu_j \leq \ell_j, \; j=-M,\ldots, M \Bigg\}
    \label{eq:rect_region_def2}.
\end{align}
Then, the following holds.

\begin{thm}[Binary Memoryless Symmetric Channel]\label{thm:Extended_Poltyrev}
For a BMS channel \eqref{eq:BMS}, the error probability of a
binary linear code with minimal Hamming distance $d_{\rm min}$ and   weight distribution $S_w, w=d_{\rm min},\dots,n$, is
bounded as follows 
\begin{align}
        &P_e \leq \sum_{\substack{\bm{\ell} \in \mathbb{Z}_{+}^{2M+1} \\  \sum_j \ell_j = n }} \left(\prod_{j=-M}^M\left(\Dsub{\varepsilon}{j|0}\right)^{\ell_{j}}\right) \min\Bigg\{\sum_{w=d_{\rm{min}}}^n S_w\nonumber \\ 
        & \quad \sum_{\substack{\bm{\mu} \in \mathcal{U}_w(\bm{\ell}) \\ \sum_{j}\mu_j\rm{LLR}_j\leq 0}} \binom{w}{\Dsub{\mu}{-M},\dots,\Dsub{\mu}{M}} \nonumber \nonumber \\ 
        & \qquad \cdot\binom{n-w}{\Dsub{\ell}{-M}-\Dsub{\mu}{-M},\dots,\Dsub{\ell}{M}-\Dsub{\mu}{M}} ,\binom{n}{\Dsub{\ell}{-M},\ldots,\Dsub{\ell}{M}}\Bigg\}, \label{eq:DMC_bound_explict}
\end{align}

\end{thm}

\begin{proof}


The error probability can be written as follows:
\begin{align}
    &P_e = P(\mathcal{E}, \mathbf{y} \in \mathcal{P}_n\mid \mathbf{0}) + P(\mathcal{E}, \mathbf{y} \in \mathcal{P}_n^c \mid \mathbf{0})\label{eq:thm:eqn100}\\
    &\leq P(\mathcal{E}, \mathbf{y} \in \mathcal{P}_n\mid \mathbf{0}) + P(\mathbf{y} \in \mathcal{P}_n^c \mid \mathbf{0})\label{eq:thm:eqn110}\\
    &=\sum_{P_n \in \mathcal{P}_n} P(\mathcal{E}, \mathbf{y} \in P_n \mid \mathbf{0} ) + P(\mathbf{y} \in \mathcal{P}_n^c \mid \mathbf{0})\label{eq:thm:eqn120}\\
    &=\sum_{P_n \in \mathcal{P}_n} P\left(\bigcup_{\mathbf{c}\in \mathcal{C}  \setminus \{\mathbf{0}\}} \mathcal{E}_c, \mathbf{y} \in P_n \middle | \mathbf{0} \right) + P(\mathbf{y} \in \mathcal{P}_n^c \mid \mathbf{0})\label{eq:thm:eqn130}\\
    &\leq \sum_{P_n \in \mathcal{P}_n} \sum_{\mathbf{c}\in \mathcal{C} \setminus \{\mathbf{0} \} } P\left(\mathcal{E}_c, \mathbf{y} \in P_n \mid \mathbf{0}\right) + P(\mathbf{y} \in \mathcal{P}_n^c \mid \mathbf{0})\label{eq:thm:eqn140}
\end{align}
where \eqref{eq:thm:eqn100} follows by considering the error event $\mathcal{E}$ jointly with the set of types $\mathcal{P}_n$ and its complement set $\mathcal{P}_n^c$, and \eqref{eq:thm:eqn110} follows by omitting the error event $\mathcal{E}$ in the second term.
Finally, \eqref{eq:thm:eqn140} follows by applying the union bound.

Defining $\mathbf{y}_{w} \triangleq y_{1}^w$ and $\mathbf{y}_{\bar{w}} \triangleq y_{w+1}^n$, and using  \eqref{eq:rect_region_def2}, we define the set all possible \emph{conditional types} of $\mathbf{y}_{w}$, given the output vector $\mathbf{y}$ is of type $P_n$ by
\begin{align}
    \mathcal{P}_w(P_n) &\triangleq \frac{1}{w}\mathcal{U}_w\Big(\bm{\ell} = nP_n\Big)  \label{eq:rect_region_def2_type},
\end{align}

Using the weight distribution $S_w$,
and applying definition \eqref{eq:rect_region_def2_type}, inequality \eqref{eq:thm:eqn140} becomes

\begin{align}
    &P_e \leq \sum_{P_n \in \mathcal{P}_n} \sum_{w=d_{\rm{min}}}^n S_w\nonumber \\ 
    &\sum_{P_w\in \mathcal{P}_w(P_n)} \hspace{-0.5cm} P\left(\mathcal{E}_c, \mathbf{y}_w \in P_w, \mathbf{y}_{\bar{w}} \in P_{n-w} \mid \mathbf{0}\right)  + 
    P(\mathbf{y} \in \mathcal{P}_n^c \mid \mathbf{0}), \label{eq:thm:eqn150}   
\end{align}
where the type $P_{n-w}$ is the ``residual type" of subvector $\mathbf{y}_{\bar{w}}$ given that $\mathbf{y}$ is of type $P_n$ and that $\mathbf{y}_w$ is of type $P_w$, i.e., $P_{n-w} = \frac{1}{n-w}(nP_n-wP_w)$.

We still need to formulate explicitly the terms  $P(\mathcal{E}_c, \mathbf{y}_w \in P_w, \mathbf{y}_{\bar{w}} \in P_{n-w} \mid \mathbf{0})$ and $P(\mathbf{y} \in \mathcal{P}_n^c \mid \mathbf{0})$.
As the channel is memoryless,  the first term can  be expressed as follows.
\begin{align}
& P\left(\mathcal{E}_c, \mathbf{y}_w \in P_w, \mathbf{y}_{\bar{w}} \in P_{n-w} \mid \mathbf{0}\right) \nonumber\\
& = P(\mathbf{y}_{\bar{w}} \in P_{n-w} \mid \mathbf{0}) P(\mathcal{E}_c,\mathbf{y}_w\in P_w \mid \mathbf{0})\label{eq:thm1:eqn200}\\
&= |T(P_{n-w})|\left(\prod_{j=-M}^M\left(\varepsilon_{_{j|0}}\right)^{(n-w) P_{n-w}(j)}\right) \mathbbm{1}\{P_w\in\mathcal{P}_w^\mathcal{E}\}\nonumber \\ 
&\quad \cdot \sum_{\mathbf{y}_w\in P_w} P_{\mathbf{Y}|\mathbf{X}}(\mathbf{y}_w|x_{w}=\mathbf{0})\label{eq:thm1:eqn210}\\
&= \mathbbm{1}\{P_w\in\mathcal{P}_w^\mathcal{E}\}\cdot |T(P_{n-w})| \prod_{j=-M}^M\left(\varepsilon_{_{j|0}}\right)^{(n-w) P_{n-w}(j)} \nonumber \\
&\quad \cdot |T(P_w)| \prod_{j=-M}^M\left(\varepsilon_{_{j|0}}\right)^{w P_{w}(j)} \label{eq:thm1:eqn220}\\
&= \mathbbm{1}\{P_w\in\mathcal{P}_w^\mathcal{E}\}\cdot |T(P_w)||T(P_{n-w})|\nonumber \\ &\quad\cdot\prod_{j=-M}^M\left(\varepsilon_{_{j|0}}\right)^{(w P_{w}(j) + (n-w)P_{n-w}(j))} \label{eq:thm1:eqn230}\\
&= \mathbbm{1}\{P_w\in\mathcal{P}_w^\mathcal{E}\}\cdot |T(P_w)||T(P_{n-w})| \prod_{j=-M}^M\left(\varepsilon_{_{j|0}}\right)^{nP_{n}(j)} \label{eq:thm1:eqn240}
\end{align}
where \eqref{eq:thm1:eqn210} follows by explicitly writing the probability of type $P_{n-w}$ and by the definition of $\mathcal{P}_w^\mathcal{E}$  \eqref{eq:error_types}, along with that of the indicator function 
\begin{align}
    \mathbbm{1}\{P_w\in\mathcal{P}_w^\mathcal{E}\} = 
    \begin{cases}
        1, &  P_w\in\mathcal{P}_w^\mathcal{E} \\
        0, &  P_w\nin\mathcal{P}_w^\mathcal{E}
    \end{cases}.\label{eq:thm1:indicator}
\end{align}
Further, \eqref{eq:thm1:eqn220} follows by explicitly writing the probability of type $P_{w}$, and \eqref{eq:thm1:eqn240} follows since $nP_{n}(j) = wP_{w}(j)+(n-w)P_{n-w}(j)$.

The term $P(\mathbf{y} \in \mathcal{P}_n^c \mid \mathbf{0})$ is explicitly given by  
\begin{align}
    P(\mathbf{y} \in \mathcal{P}_n^c \mid \mathbf{0}) &= \sum_{P_n \in \mathcal{P}_n^c} |T(P_n)|\prod_{j=-M}^M\left(\varepsilon_{_{j|0}}\right)^{n P_n(j)}. \label{eq:thm1:eqn250}
\end{align}


Thus, substituting \eqref{eq:thm1:eqn240} and \eqref{eq:thm1:eqn250} and the size of $|T(P_n)|$ into \eqref{eq:thm:eqn150}, we obtain
\begin{align}
    P_e &\leq \sum_{P_n \in \mathbbm{P}_n} \left(\prod_{j=-M}^M\left(\varepsilon_{_{j|0}}\right)^{nP_n(j)}\right)\nonumber \\
    &\; \cdot\min\left\{\sum_{w=d_{\rm{min}}}^n S_w \sum_{P_w\in \mathcal{P}_w(P_n)} g(P_w,P_{n-w}), \binom{n}{\Dsub{\ell}{-M},\ldots,\Dsub{\ell}{M}} \right\},\label{eq:DMC_bound} 
\end{align}
where we have defined
\begin{align}
    & g(P_w,P_{n-w})  \triangleq \mathbbm{1}\{P_w\in\mathcal{P}_w^\mathcal{E}\}\cdot |T(P_w)||T(P_{n-w})| \\
    &=\mathbbm{1}\{P_w\in\mathcal{P}_w^\mathcal{E}\}\cdot \binom{w}{\Dsub{\mu}{-M},\dots,\Dsub{\mu}{M}} \binom{n-w}{\Dsub{\ell}{-M}-\Dsub{\mu}{-M},\dots,\Dsub{\ell}{M}-\Dsub{\mu}{M}}
    .  \label{eq:g1_def}
\end{align}
Substituting the explicit characterization of $\mathcal{P}_w^\mathcal{E}$ \eqref{eq:error_types_LLR}  into \eqref{eq:g1_def}, and the latter into 
 into \eqref{eq:DMC_bound},
yields \eqref{eq:DMC_bound_explict}.


    
\end{proof}

\begin{remark}\label{rem1}
    The computational complexity of the bound is polynomial in the block length, specifically $O(n^{2(|\mathcal{Y}|+1)})$ where $|\mathcal{Y}|$ is the output alphabet size.
\end{remark}

\section{Channel Classes with Explicit  Characterization of Error Event Regions}\label{sec:specific_channel}


While in general, the pairwise error event region \eqref{eq:pairwise_error_event_region} varies with the channel parameters, for certain parametric channel classes of interest, this region is independent of the channel parameter and admits a simple characterization.

Similarly, while the ``large-noise" region may be taken to be any set, in order to obtain a tight bound one needs to choose the region wisely, i.e., taylored to the channel at hand. For certain channel classes, there is a natural candidate set for the region as we now demonstrate.

\subsection{Explicit Characterization for the Class of BSC} \label{sec:specific_channel_BSC}
Consider the class of channels BSC($\varepsilon$), illustrated in
\figref{fig:BSC_model}.


\begin{figure}[htbp]
    \centering
        \includegraphics[scale=0.25
        ]{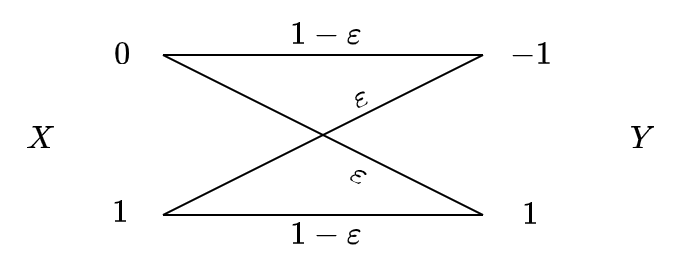}
        \caption{Binary symmetric channel - BSC($\varepsilon$).}
        \label{fig:BSC_model}
\end{figure}

Following \thmref{thm:Extended_Poltyrev}, for BSC($\varepsilon$), the set of all possible non-normalized types of output vectors is given by $n\mathbbm{P}_n=\{(n-\ell_1,\ell_1): 0,\leq \ell_1 \leq n\}$, since $\ell_{-1} = n-\ell_1$.
Additionally the set $\mathcal{U}_w(\ell_1)$ is given by $\mathcal{U}_w(\ell_1)=\{(w-\mu_1,\mu_1): 0\leq \mu_1\leq \min\{\ell_1,w\},\; w-\mu_1\leq n-\ell_1\}$. Further, for BSC($\varepsilon$), since ${\rm LLR}_{-1}=-{\rm LLR}_{1}$, the inequality $\mu_{-1}{\rm LLR}_{-1}+\mu_{1}{\rm LLR}_{1}\leq 0$ becomes $\mu_{1}\geq w/2$, since $\mu_{-1} = n-\mu_{1}$, independent of $\varepsilon$.

Therefore, with the above specification of $n\mathbbm{P}_n$ and $\mathcal{U}_w(\ell_1)$, and applying $\mu_{1}\geq w/2$ in \eqref{eq:DMC_bound_explict}, we obtain \eqref{eq:BSC_Poltyrev_bound}, the well-known Poltyrev bound for the BSC.

\subsection{Explicit Characterization for the Class of BEC}\label{sec:specific_channel_BEC}
Consider the class of channels BEC($\delta$) channels, 
illustrated in 
\figref{fig:BEC_model}.

\begin{figure}[htbp]
    \centering
        \includegraphics[scale=0.25
        ]{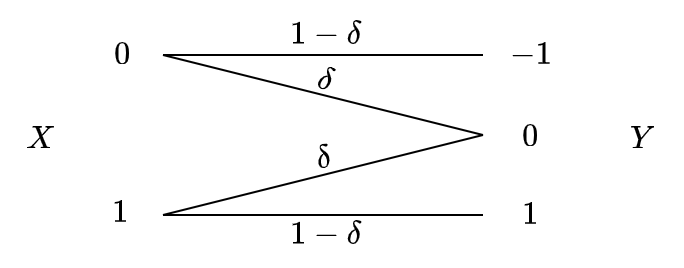}
        \caption{Binary erasure channel - BEC($\delta$).}
        \label{fig:BEC_model}
\end{figure}

Following \thmref{thm:Extended_Poltyrev}, for BEC($\delta$), the set of all possible non-normalized output types is given by $n\mathbbm{P}_n=\{(n-\ell_0-\ell_1,\ell_0,\ell_1):\ell_0\geq 0,\ell_1\geq 0, 0,\leq \ell_0+\ell_1 \leq n\}$. Since $\Dsub{\varepsilon}{1|0} = 0$, we get that $\Dsub{\varepsilon}{1|0}^{\ell_1} = 0$ for $\ell_1 > 0$, and $\Dsub{\varepsilon}{1|0}^{\ell_1} = 1$ for  $\ell_1 = 0$.\footnote{We have used the convention that $0^0=1$.} Thus, we may set $\ell_1 =0$ in  \eqref{eq:DMC_bound_explict}, reducing the summation to be over a single parameter $\ell$. Furthermore, the set $\mathcal{U}_w(\ell)$ amounts to $w$  erasures (all elements of the sub-vector being erased), yielding

\begin{lem}[BEC]
For a memoryless binary erasure channel, i.e., BEC($\delta$),  the error probability is
bounded by
\begin{align}
    P_e &\leq \sum_{\ell=0}^{n} \delta^{\ell}(1-\delta)^{n-\ell}\min\left\{\sum_{w=d_{\rm{min}}}^{n} S_w  \binom{n-w}{\ell-w}, \binom{n}{\ell} \right\}. 
\end{align}
\end{lem}

\subsection{Explicit Characterization for Hybrid BSC-BEC }\label{sec:specific_channel_BSC-BEC}
Consider 
a hybrid BSC-BEC channel as illustrated in 
\figref{fig:BSC_BEC},
where $\varepsilon\in [0,1/2]$ and $\delta\in[0,1]$ such that $\varepsilon+\delta\leq 1 $. 

\begin{figure}[htbp]
    \centering
        \includegraphics[scale=0.3
        ]{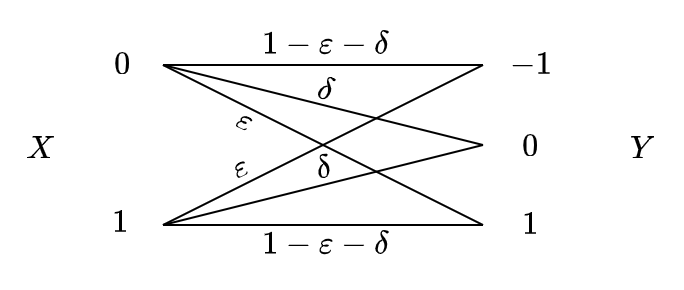}
        \caption{Hybrid BSC-BEC($\varepsilon,\delta$)  channel.}
        \label{fig:BSC_BEC}
\end{figure}

The set $n\mathbbm{P}_n$ is now given by 
\begin{align}
n\mathbbm{P}_n=\{(n-\ell-u,u,\ell): 0\leq \ell \leq n,0 \leq u \leq n - \ell\},\label{eq:BSC-BEC_Un}
\end{align}
where $u$ represents the number of erased bits ($0\rightarrow 0, 1 \rightarrow 0$) and $\ell$ represents the number of flipped bits ($0\rightarrow 1,1\rightarrow -1$). In addition, the set $\mathcal{U}_w(\ell,u)$ is given by 
\begin{align}
\mathcal{U}_w(\ell,u)&=\left\{(w-\rho-\mu,\rho,\mu): 0\leq \mu \leq \min(\ell,w),\nonumber\right. \\ &\left. \qquad\qquad\qquad 0 \leq \rho \leq \min(u,w), \mu+\rho \leq w\right\}.\label{eq:BSC-BEC_Uw}     
\end{align}

We substitute $\ell_{-1} \rightarrow n-\ell-u,\ell_0 \rightarrow u, \ell_1 \rightarrow \ell$ and $\mu_{-1} \rightarrow w-\rho-\mu, \mu_0 \rightarrow \rho, \mu_1 \rightarrow \mu$ in \eqref{eq:DMC_bound_explict}, using the definitions of $\mathbbm{P}_n$, $\mathcal{U}_w(\ell,u)$ from \eqref{eq:BSC-BEC_Un}, \eqref{eq:BSC-BEC_Uw}, respectively. Furthermore, with this substition, the condition $\sum_{j=-M}^M\mu_j\rm{LLR}_j\leq 0$ in \eqref{eq:DMC_bound_explict} becomes 
$\mu\geq (w-\rho)/2$. This follows from the fact that $M=1$,  ${\rm LLR}_{-1}=-{\rm LLR}_{1}$ and ${\rm LLR}_{0} = 0$. Consequently, we obtain the following bound for the hybrid BSC-BEC channel.

\begin{lem}[BSC-BEC]
For a BMS channel \eqref{eq:BMS} and $M=1$, i.e., BSC-BEC($\varepsilon,\delta$), the error probability is bounded by
\begin{align}
    P_e &\leq \sum_{\ell=0}^{n}\sum_{u=0}^{n-\ell} \varepsilon^{\ell}\delta^{u}(1-\varepsilon-\delta)^{n-\ell-u}\nonumber\\
    &\quad\cdot\min\left\{\sum_{w=d_{\rm{min}}}^{n} S_w 
    \sum_{\rho=0}^{\min\{u,w\}}
    \sum_{\mu=\frac{w-\rho}{2}}^{\min\{\ell,w\}} 
    \binom{w}{\mu,\rho,w-\mu-\rho} \right.\nonumber \\ 
    &\qquad\qquad \left. \cdot \binom{n-w}{\ell-\mu,u-\rho,n-w-(\ell-\mu)-(u-\rho)}, \right.\nonumber \\
    &\qquad\qquad\qquad \left. \binom{n}{\ell,u,n-\ell-u} \right\}.\label{eq:BSC-BEC_bound} 
\end{align}
\end{lem}

\subsection{Explicit Characterization for a Class of Quinary Channels}\label{sec:specific_channel_quinary}
 Consider a BMS channel as described in \figref{fig:Quinary_channel}, where $\varepsilon,\delta,\gamma \in [0,1]$, such that $\varepsilon+\delta + \gamma\leq 1 $,
We denote the channel by Quinary($\varepsilon,\delta, \gamma$).


\begin{figure}[htbp]
    \centering
        \includegraphics[scale=0.25
        ]{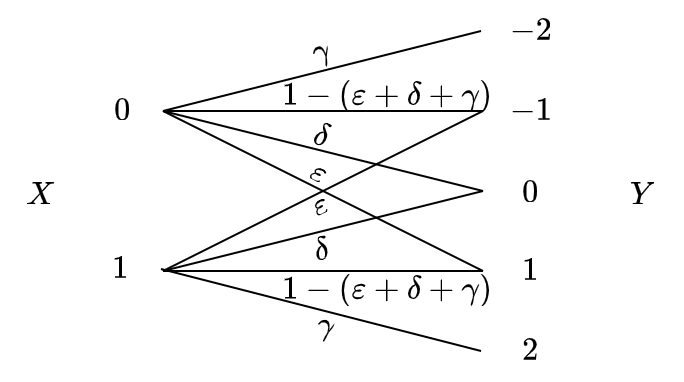}
        \caption{Quinary($\varepsilon,\delta, \gamma$) channel.}
        \label{fig:Quinary_channel}
\end{figure}

We may interpret the channel output as follows: $y=\pm 2$ represents a ``strong-correct" decision (the channel input is known with certainty), $y=\pm 1$ represents a ``weak-correct" decision (the channel input can be guessed with a probability of success greater than half),
and $y=0$ represents an erasure (the input values are equally likely). 

These levels of certainty 
arise in flash memory when reading a page 
 \cite{InsideNAND_book} where, in practice,  
a strong-correct decision has a very small probability of error. In our model, this probability is zero, as depicted in \figref{fig:Quinary_channel}. 
Furthermore, in successive decoding of flash memory pages,
the certainty of a strong correct becomes even more distinct.


The pairwise error-event region $\mathcal{P}_w^\mathcal{E}$, as defined in \eqref{eq:error_types_LLR}, now contains all the vectors in $P_w$ that satisfy $2\mu + \rho \geq w$ \emph{and} $\nu = 0$, where $\nu$, $\rho$, and $\mu$ represent the numbers of strong-corrects, weak-corrects and erasures, respectively. The large-noise region is the complement of a Hamming sphere as in the previous examples. Both sets are independent of the parameters $(\varepsilon, \delta, \gamma)$.

The bound for the Quinary($\varepsilon,\delta, \gamma$) channel is stated without proof. 
\begin{lem}[Quinary channel]
For a BMS channel \eqref{eq:BMS} with quinary output, i.e., Quinary($\varepsilon,\delta,\gamma$), the error probability is bounded by
\begin{align}
    P_e &\leq \sum_{\bm{\ell}\in \mathbb{Z}_{+}^n:\sum_j\ell_j =n} \varepsilon^{\ell_1}\delta^{\ell_2}\gamma^{\ell_3}(1-\varepsilon-\delta-\gamma)^{\ell_0} \nonumber\\
    &\quad\cdot\min\left\{\sum_{w=d_{\rm{min}}}^{n} S_w 
    \sum_{\substack{\mu,\rho \\ 2\mu+\rho\geq w}}
    \binom{w}{\mu,\rho,n-\mu-\rho} \right.\nonumber \\ 
    &\qquad \left. \cdot \binom{n-w}{\ell_1-\mu,\ell_2-\rho,n-w-(\ell_1-\mu)-(\ell_2-\rho)}, \right.\nonumber \\
    &\qquad\qquad\qquad \left. \binom{n}{\ell_0,\ell_1,\ell_2,\ell_3} \right\}.\label{eq:Quinary_bound} 
\end{align}
\end{lem}


\section{Computationally Simplified Bound} \label{sec:rect_bounds}

By Remark~\ref{rem1}, the computational complexity of bound \eqref{eq:DMC_bound_explict} is $O(n^{2(|\mathcal{Y}|+1)})$. Unfortunately, when the block length is large, and even more so, when the output alphabet is large, this translates to high computational complexity. 
We now reduce the computational complexity of the bound to linear complexity, i.e., $O(n)$, at the cost of  inducing some loss in  tightness. We do so by limiting the search over types, from ranging over the full set $\mathbbm{P}_n$, to a restricted ``rectangular set of types'' defined next.

Let $\mathbf{m}=(\Dsub{m}{-M+1},\ldots \Dsub{m}{M})$ be a length-$|\mathcal{Y}|-1$ vector, the $i$-th entry of which restricts the range of the occurrence of the $i$-th symbol, for all $i$ \emph{except for} $i=-M$. Then, the rectangular set of types is defined as follows. 


\begin{align}
\tilde{\mathcal{U}}_n(\mathbf{m}) &\triangleq \Bigg\{ \bm{\ell} : \sum_{i=-M}^M \ell_i=n,\; 
{\ell_i \in\mathbb{Z}_{+}, \; i=-M,\ldots, M}
\nonumber \\
    &\qquad , 0\leq \ell_j \leq m_j, \; j=-M+1,\ldots, M \Bigg\}.\label{eq:rect_region_minus_firsr_def}
\end{align}

Restricting the type search in Theorem~\ref{thm:Extended_Poltyrev} to the set $\tilde{\mathcal{U}}_n(\mathbf{m})$ results in the following bound.

\begin{cor}[Rectangle-type search region]
For a BMS channel \eqref{eq:BMS}, and $\mathbf{m}=(\Dsub{m}{-M+1},\ldots \Dsub{m}{M}), m_i\in\mathbb{Z}_{+}$, the error probability is bounded by
\begin{align}
    &P_e(\mathbf{m}) \leq \sum_{\bm{\ell} \in \tilde{\mathcal{U}}_n(\bm{m)}} \left(\prod_{j=-M}^M[\Dsub{\varepsilon}{j|0}]^{\ell_{j}}\right) 
    \min\left\{\sum_{w=d_{\rm{min}}}^n S_w \right.\nonumber\\ 
    & \left. \sum_{\substack{\bm{\mu} \in \mathcal{U}_w(\bm{\ell}) \\   \sum_{j=-M}^M\mu_j\rm{LLR}_j\leq 0}} \binom{w}{\Dsub{\mu}{-M},\dots,\Dsub{\mu}{M}} \binom{n-w}{\Dsub{\ell}{-M}-\Dsub{\mu}{-M},\dots,\Dsub{\ell}{M}-\Dsub{\mu}{M}},\right.\nonumber \\
    & \left. \binom{n}{\Dsub{\ell}{-M},\ldots,\Dsub{\ell}{M}}\right\} +  \sum_{\bm{\ell} \nin \tilde{\mathcal{U}}_n(\bm{m)}} \left(\prod_{j=-M}^M\left(\Dsub{\varepsilon}{j|0}\right)^{\ell_{j}}\right)\binom{n}{\Dsub{\ell}{-M},\ldots,\Dsub{\ell}{M}}. \label{eq:DMC_bound_rec_search}        
\end{align}
\end{cor}

As a result of this bounding technique, the bound is divided into three regions: the first is the inner region of types where the error events are the union of all pairwise error events; the second is the region inside the predefined rectangular region but outside the inner region; and the third contains the types outside the rectangular region. Thus, 
the union of the second and third regions constitutes the ``large-noise region" of the bound.

A looser bound can be derived by upper bounding the second term in \eqref{eq:DMC_bound_rec_search}, i.e., the probability that the output vector falls in the ``large-noise region''. Specifically, 
this probability may be bounded by applying the Chernoff bound for the event that $\frac{m_j}{n} > \varepsilon_{j|0}$, where $j=-M+1,\ldots,M$, and then applying a union bound over the different values of  $j$, resulting in

\begin{align}
& P_e(\mathbf{m}) \leq \sum_{\bm{\ell} \in \tilde{\mathcal{U}}_n(\bm{m)}} \left(\prod_{j=-M}^M\left(\Dsub{\varepsilon}{j|0}\right)^{\ell_{j}}\right) \min\Bigg\{\sum_{w=d_{\rm{min}}}^n S_w  \nonumber \\
& \sum_{\substack{\bm{\mu} \in \mathcal{U}_w(\bm{\ell})\\ \sum_{j=-M}^M\mu_j\rm{LLR}_j \leq 0}} \binom{w}{\Dsub{\mu}{-M},\dots,\Dsub{\mu}{M}} \binom{n-w}{\Dsub{\ell}{-M}-\Dsub{\mu}{-M},\dots,\Dsub{\ell}{M}-\Dsub{\mu}{M}}\Bigg\} \nonumber \\ 
&\qquad\qquad\qquad\qquad + \sum_{j=-M+1}^{M}\Exp{-n D\left( \frac{m_j}{n} \middle\| \Dsub{\varepsilon}{j|0} \right) }. \label{eq:DMC_bound_rec_sel_Chernoff}        
\end{align}

\section{Numerical Results}\label{sec:numerical}

The bounds derived in this work—
the general bound of Theorem~ \ref{thm:Extended_Poltyrev}, i.e., 
\eqref{eq:DMC_bound_explict}, 
and its reduced-complexity corollary, stated in  
\eqref{eq:DMC_bound_rec_search}, 
are illustrated by applying them 
to the hybrid BSC-BEC  channel defined in \secref{sec:specific_channel_BSC-BEC}.

The bounds were evaluated for  BCH codes of several rates, with a block length of $n=127$, as well as for a binomial Hamming weight distribution, i.e., the weight distribution of random linear codes. The weight distributions for the BCH codes are reported in \cite{kasami1966weight, MacSloane, berlekampalgebraic_1968}.

The results also compared with
the Shulman-Feder (SF) bound \cite{ShulmanFederBound} and 
the Miller-Burshtein (MB) bound that are also applicable to linear codes with a known spectrum.\footnote{The MB bound  is obtained by finding the optimal partitioning of the sets $U$ and $U^c$ in \cite[Thm. 1]{ldpcens}, which turns out to require a complexity of only $O(n\log(n))$.}
As a benchmark, we also plot the random-coding Gallager error exponent  \cite{Gallager68}, evaluated  for the parameters chosen. 
%

\begin{figure}[htbp]
    \centering
        \includegraphics[scale=0.6
        ]{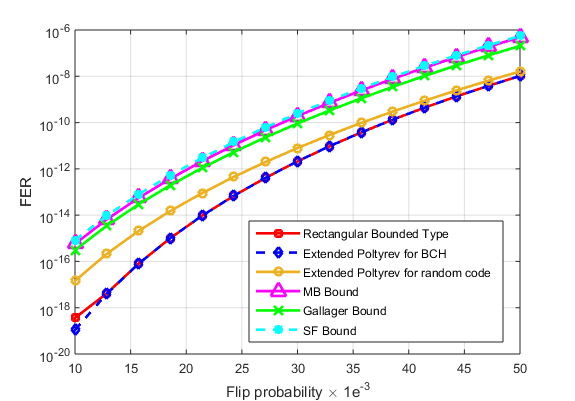}
        \caption{The Frame Error Rate (FER) of the hybrid BSC-BEC channel for BCH($n=127, k=29, d_{\rm{min}}=43$) with a code rate of $R=0.23$ and an erasure probability of $\delta = 0.1$.}
        \label{fig:BCH_127_29_43}
\end{figure}

\begin{figure}[htbp]
    \centering
        \includegraphics[scale=0.6
        ]{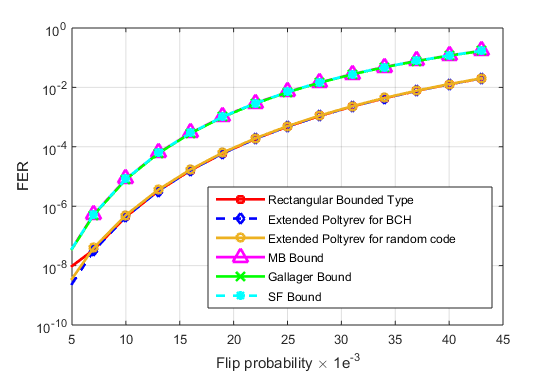}
        \caption{The Frame Error Rate (FER) of the hybrid BSC-BEC channel for BCH($n=127, k=64, d_{\rm{min}}=21$) with a code rate of $R=0.503$ and an erasure probability of $\delta = 0.1$.}
        \label{fig:BCH_127_64_21}
\end{figure}

\begin{figure}[htbp]
    \centering
        \includegraphics[scale=0.6
        ]{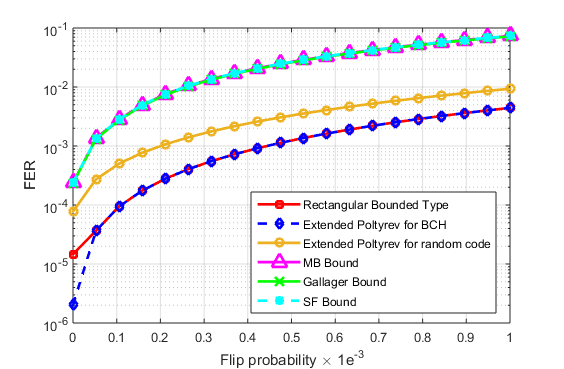}
        \caption{The Frame Error Rate (FER) of the hybrid BSC-BEC  channel for BCH($n=127, k=113, d_{\rm{min}}=5$) with a code rate of $R=0.89$ and an erasure probability of $\delta = 0.01$.}
        \label{fig:BCH_127_113_5}
\end{figure}

In \figref{fig:BCH_127_29_43}, \ref{fig:BCH_127_64_21} and \ref{fig:BCH_127_113_5}, the Frame Error Rate (FER) is plotted as a function the crossover probability, due to the BSC part, at a fixed erasure probability, for BCH(127,29,43), BCH(127,64,21), and BCH(127,113,5), respectively. As can be seen, the extended Poltyrev bound \eqref{eq:BSC-BEC_bound}  is significantly tighter than the MB and SF bounds. However, the computation of the bound \eqref{eq:BSC-BEC_bound} is quite demanding. Nonetheless, the 
bounded-type search bound
\eqref{eq:DMC_bound_rec_search} 
is rather close to the extended Poltyrev bound, in this example, 
with linear computational complexity in the block length $n$.

It is interesting to note that, taking the extended Poltyrev bound as a figure of merit, the BCH code exhibits significantly better performance as compared to a random code. 
This is due to the fact that, for the moderate block length  $n=127$, the BCH code has a larger minimum distance than the Gilbert-Varshamov bound \cite{Gilbert52,Varshamov57}.


\section{Discussion}\label{sec:discussion}
This work extended the Poltyrev bound to general binary memoryless symmetric discrete-output channels using the method of types. 
The bound is of practical relevance to scenarios where the output alphabet size is small and/or the block length is small to moderate. Several numerical examples were presented in which the bound is significantly tighter than previously reported bounds. A reduced-complexity variant of the bound was also derived.




\bibliographystyle{IEEEtran}
\bibliography{tal1}

\begin{thebibliography}{10}
\providecommand{\url}[1]{#1}
\csname url@samestyle\endcsname
\providecommand{\newblock}{\relax}
\providecommand{\bibinfo}[2]{#2}
\providecommand{\BIBentrySTDinterwordspacing}{\spaceskip=0pt\relax}
\providecommand{\BIBentryALTinterwordstretchfactor}{4}
\providecommand{\BIBentryALTinterwordspacing}{\spaceskip=\fontdimen2\font plus
\BIBentryALTinterwordstretchfactor\fontdimen3\font minus \fontdimen4\font\relax}
\providecommand{\BIBforeignlanguage}[2]{{%
\expandafter\ifx\csname l@#1\endcsname\relax
\typeout{** WARNING: IEEEtran.bst: No hyphenation pattern has been}%
\typeout{** loaded for the language `#1'. Using the pattern for}%
\typeout{** the default language instead.}%
\else
\language=\csname l@#1\endcsname
\fi
#2}}
\providecommand{\BIBdecl}{\relax}
\BIBdecl

\bibitem{gallager1963low}
R.~G. Gallager, ``Low density parity check codes, monograph,'' 1963.

\bibitem{shamai2002variations}
S.~Shamai and I.~Sason, ``Variations on the gallager bounds, connections, and applications,'' \emph{IEEE Transactions on Information Theory}, vol.~48, no.~12, pp. 3029--3051, 2002.

\bibitem{hof2009performance}
E.~Hof, I.~Sason, and S.~Shamai, ``Performance bounds for nonbinary linear block codes over memoryless symmetric channels,'' \emph{IEEE transactions on information theory}, vol.~55, no.~3, pp. 977--996, 2009.

\bibitem{PoltyrevBound94}
G.~Poltyrev, ``Bounds on the decoding error probability of binary linear codes via their spectra,'' \emph{IEEE Transactions on Information Theory}, vol.~40, no.~4, pp. 1284--1292, 1994.

\bibitem{InsideNAND_book}
R.~Micheloni, L.~Crippa, and A.~Marelli, \emph{Inside NAND Flash Memories}, 2010.

\bibitem{3DFlash_book}
R.~Micheloni, \emph{3D Flash Memories}.\hskip 1em plus 0.5em minus 0.4em\relax Springer Netherlands, 2016.

\bibitem{Kurkoski_TIT_2014}
B.~M. Kurkoski and H.~Yagi, ``Quantization of binary-input discrete memoryless channels,'' \emph{IEEE Transactions on Information Theory}, vol.~60, no.~8, pp. 4544--4552, 2014.

\bibitem{CoverBook2Edition}
T.~M. Cover and J.~A. Thomas, \emph{Elements of Information Theory}, 2nd~ed.\hskip 1em plus 0.5em minus 0.4em\relax New York: Wiley, 2006.

\bibitem{kasami1966weight}
T.~Kasami, ``Weight distributions of {B}ose-{C}haudhuri-{H}ocquenghem codes,'' \emph{Coordinated Science Laboratory Report no. R-317}, 1966.

\bibitem{MacSloane}
F.~J. MacWilliams and N.~J.~A. Sloane, \emph{The Theory of Error Correcting Codes}.\hskip 1em plus 0.5em minus 0.4em\relax Elsevier, 1977.

\bibitem{berlekampalgebraic_1968}
E.~R. Berlekamp, \emph{Algebraic coding theory}.\hskip 1em plus 0.5em minus 0.4em\relax McGraw-Hill, 1968.

\bibitem{ShulmanFederBound}
N.~Shulman and M.~Feder, ``Random coding techniques for nonrandom codes,'' \emph{IEEE Transactions on Information Theory}, vol.~45, pp. 2101--2104, Sep. 2002.

\bibitem{ldpcens}
G.~Miller and D.~Burshtein, ``Bounds on the maximum likelihood decoding error probability of low density parity check codes,'' \emph{IEEE Transactions on Information Theory}, vol. IT-47, pp. 2696--2710, Nov. 2001.

\bibitem{Gallager68}
R.~G. Gallager, \emph{Information {T}heory and {R}eliable {C}ommunication}.\hskip 1em plus 0.5em minus 0.4em\relax New York, N.Y.: Wiley, 1968.

\bibitem{Gilbert52}
E.~N. Gilbert, ``A comparison of signalling alphabets,'' \emph{Bell Labs Technical Journal}, vol.~31, pp. 504--522, 1952.

\bibitem{Varshamov57}
R.~R. Varshamov, ``Estimate of the number of signals in error correcting codes,'' \emph{Dokl. Akad. Nauk SSSR}, vol. 117, no.~5, pp. 739--741, 1957.

\end{thebibliography}

\end{document}